\documentclass[journal]{IEEEtran}
\usepackage{amsmath,amsthm,amssymb,paralist,subfigure,cite,graphicx,color}
\usepackage[ruled,vlined,linesnumbered]{algorithm2e}
\usepackage{flushend}
\usepackage{hyperref}

\ifCLASSINFOpdf
\else
\usepackage[dvips]{graphicx}
\fi
\usepackage{url}

\usepackage{graphicx}

\newtheorem{lem}{Lemma}
\newtheorem{ass}{Assumption}
\newtheorem{theorem}{Theorem}

\def\mb{\mathbf}

\def\mc{\mathcal}

\def\mb{\mathbf}

\def\mc{\mathcal}

\newcommand{\ab}[1]{\textcolor{blue}{{#1}}}

\begin{document}
		\title{ Distributed Energy Resource Management: All-Time Resource-Demand Feasibility, Delay-Tolerance, Nonlinearity, and Beyond}
		\author{Mohammadreza Doostmohammadian
		\thanks{Mohammadreza Doostmohammadian is with the Faculty of Mechanical Engineering at Semnan University, Semnan, Iran, email: \texttt{doost@semnan.ac.ir}. He was also with the  School of Electrical Engineering and Automation at Aalto University, Espoo, Finland, email: \texttt{name.surname@aalto.fi}. 
		}}
		\maketitle
	
	\begin{abstract}
		In this work, we propose distributed and networked energy management scenarios to optimize the production and reservation of energy among a set of distributed energy nodes. In other words, the idea is to optimally allocate the generated and reserved powers based on nodes' local cost gradient information while meeting the demand energy. One main concern is the all-time (or anytime) resource-demand feasibility, implying that at all iterations of the scheduling algorithm, the balance between the produced power and demand plus reserved power must hold. The other concern is to design algorithms to tolerate communication time-delays and changes in the network. Further, one can incorporate possible model nonlinearity in the algorithm to address both inherent (e.g., saturation and quantization) and purposefully-added (e.g., signum-based) nonlinearities in the model. The proposed optimal allocation algorithm addresses all the above concerns, while it benefits from possible features of the distributed (or networked) solutions such as no-single-node-of-failure and distributed information processing. We show both the all-time feasibility of the proposed scheme and its convergence under certain bound on the step-rate using Lyapunov-type proofs.  
	\end{abstract}
	
	\IEEEpeerreviewmaketitle
	
	\section{Introduction} \label{sec_intro}
	\IEEEPARstart{D}{istributed} \ab{energy resource management (DERM) and networked resource allocation strategies intend to coordinate and optimize the distributed energy resources in a smart grid \cite{Alanne2006distributed,strezoski2023distributed}.} Energy resources typically include small-scale renewable energy sources like solar photovoltaic (PV) systems, wind turbines, and energy storage systems such as batteries. DERM focuses on integrating and managing these decentralized energy resources to enhance the efficiency, reliability, and flexibility of the overall grid system. It involves monitoring, controlling, and optimizing the operation of individual resources and coordinating their collective actions to meet the demands of the smart grid while optimizing energy costs. DERM platforms utilize advanced control algorithms and optimization techniques to ensure the efficient operation of resources. They can balance energy supply and demand, optimize energy dispatch, and manage energy storage systems to minimize grid resource-demand imbalances and economically optimize the utilization of energy resources (e.g., renewable resources). These schemes also benefit from distributed/networked processing of data and information.
	
	\subsection{Mathematical Formulation}
	Distributed and networked equality-constraint optimization algorithms find application in different resource allocation scenarios from coverage control \cite{MSC09} and CPU-scheduling \cite{rikos2021optimal,grammenos2023cpu,ccta_cpu} to energy resource management \cite{kar_edp,cherukuri2015tcns,4497237,7042749}. In the latter, the optimization problem is defined as minimizing the sum of local energy cost functions:
		\begin{align} \label{eq_dra}
			\min_{\mb{x},\mb{y}}
		~~ & F(\mb{x},\mb{y}) = \sum_{i=1}^{N} f_i(x_i) + \sum_{i=1}^{m} e_i(y_i),
	\end{align}
	with  $x_i \in \mathbb{R}$ as the power state of the energy-production node $i \in \{1,\dots,N\}$, $y_i \in \mathbb{R}$ as the power state of the energy-reserving node $i \in \{1,\dots,m\}$,
	vectors $\mb{x} = [x_1;\dots;x_N] \in \mathbb{R}^n$ and $\mb{y} = [y_1;\dots;y_m] \in \mathbb{R}^m$ as the global vector states,  $f_i:\mathbb{R} \mapsto \mathbb{R}$ as the local cost function at energy-production node $i$, $e_i:\mathbb{R} \mapsto \mathbb{R}$ as the local strictly-convex cost function at energy-reserving node $i$, and $F:\mathbb{R}^n \mapsto \mathbb{R}$ as the overall cost. The resource-demand feasibility equality-constraint is,
	\begin{align} \label{eq_b}
	\sum_{i=1}^N x_i  = b + \sum_{i=1}^m y_i,
    \end{align}	
	which implies that the sum of the generated powers equals the demand  $b \in \mathbb{R}$ plus the sum of the reserved powers in battery-type nodes. This constraint is called \textit{resource-demand feasibility}. In energy management it is key to keep this constraint satisfied at all time-instants; this is because at all time-instants the balance between the generated power, reserved power, and the demand power needs to be preserved, otherwise, it may cause service disruption. 
	
	The energy nodes are further constrained with the so-called box constraints limiting the range of powers as,
\ab{	\begin{align} \label{eq_box}
	\underline{d}_x \leq x_i \leq \overline{d}_x,~\underline{d}_y \leq y_i \leq \overline{d}_y
    \end{align}	}
 	
	\subsection{Review of Literature and Contributions}
	The preliminary solutions are mostly linear 
	\cite{cherukuri2015tcns,boyd2006optimal,doan2017scl} over all-time connected networks. The linear solutions cannot address possible link or node nonlinearity due to, e.g., log quantization or saturation. \ab{For example, the link nonlinearity may represent data-quantization, where the information sent from one generator node to another generator node is quantized. The node nonlinearity, for example, may represent the saturation function to model ramp-rate-limits at the generators in the automatic-generation-control setup. This issue implies that the rate of change in the produced power by the generators is bounded and cannot follow any high rate.} In our proposed solution it is even possible to add sgn-based  nonlinearity \ab{to improve the convergence rate or robustness to impulsive-noise. Such sgn-based dynamics can reach convergence in finite or fixed-time scenarios \cite{shang2017finite,taes2020finite,polyakov2011nonlinear,song2021fixed}. In general, the proposed solution in this work addresses any nonlinearity that can be modelled by sign-preserving odd mapping, e.g., log quantization.}
	
	There are dual-optimal solutions based on alternating-direction-method-of-multipliers (ADMM).
	The existing ADMM-based solutions \cite{wei_me,banjac2019decentralized,falsone2020tracking,carli2019distributed,cdc22} converge to feasible solution \textit{asymptotically} and, thus, fail to preserve resource-demand feasibility at all times. This implies that along the solution dynamics and in case of terminating the algorithm before asymptotic convergence, the resource-demand feasibility does not necessarily hold which may cause service disruption in energy management. Further, the existing ADMM-based solutions \cite{wei_me,banjac2019decentralized,falsone2020tracking,carli2019distributed} require all-time network connectivity with no time-delay in contrast to uniform connectivity subject to heterogeneous latency in this work. Moreover, the node and link nonlinearity cannot be addressed by the ADMM methods and other Lagrangian-based methods \cite{doan2017ccta}. This motivates the gradient-based solution in this work. 
	
	Motivated by distributed \textit{consensus-based} algorithms subject to communication latency and possible time-delays in data-transmission  \cite{lcss21,hadjicostis2013average,grammenos2023cpu,ballotta2023can,SensNets:Olfati04,aryankia2020neuro}, we design algorithms to tolerate certain bounded time-delays over the communication network. We consider the general possible form of the latency, i.e., the assumption on the delays is that they are bounded and the time-delays are arbitrary, random, time-varying, and heterogeneous at different links. The assumption on the boundedness of the delays is to ensure that the messages from one node ultimately reach their destination node and there is no packet drop\footnote{The case of packet drop over the network is left for our future research direction as discussed in Section~\ref{sec_con}}. \ab{There are some literature on distributed algorithms under communication time-delay \cite{li2019bandit,gatsis2012asynchronous,marques2010stochastic}. However, these works are mainly dual-based formulation and do not address all-time feasibility and link/node nonlinearity. Note that the current work addresses latency, feasibility at all iterations, uniform-connectivity, and handles nonlinear models at links/nodes all-together. Recall that, for example, violating all-time feasibility implies some gap in resource-demand balance that may cause service disruption. One main application of this work is DC optimal power flow in distributed setup \cite{disfani2015distributed,kargarian2016toward,xu2018distributed,biagioni2020learning}. What the current paper adds to the existing literature \cite{disfani2015distributed,kargarian2016toward,xu2018distributed,biagioni2020learning} is addressing latency, model nonlinearity, and no feasibility-gap along the solution. In these aspects this paper advances the state-of-the-art. For example, these power system literature cannot address ramp-rate-limits in the automatic-generation-control which refers to the saturated rate of generated power or the quantized data which refers to the quantization of information shared over the communication network.}   
	
	\textit{Paper Organization:} Section~\ref{sec_frame} reformulates the problem in a compact form. Section~\ref{sec_sol} proposes our main distributed solution proposed to solve the problem. Section~\ref{sec_delay} discusses the solution under possible communication time-delay. Section~\ref{sec_sim} provides the simulations, and finally, Section~\ref{sec_con} concludes the paper.
	 
	\textit{General Notations:} All ones vector of size $n$ is denoted by $\mb{1}_n$. $I_n$ denotes the identity matrix of size $n$. The operator ``$;$'' denotes column concatenation. $\partial_x$ denotes the derivative with respect to $x$, i.e., $\frac{d}{dx}$. 
	
	\section{Problem Reformulation} \label{sec_frame}
	We define new compact global variable $\mb{z} = [\mb{x};\mb{y}] \in \mathbb{R}^n$ with $n = N+m$ and the cost function~\eqref{eq_dra} is reformulated in compact form as,
	\begin{align} \label{eq_dra2}
			\min_{\mb{z}}
		~~ & H(\mb{z}) = \sum_{i=1}^{n} h_i(z_i),
	\end{align}
    \ab{This cost function is not necessarily quadratic and any non-quadratic model to address different problems can be considered as the cost. The optimization formulation \eqref{eq_dra2}, in general, may address different problems including automatic generation control, economic dispatch, and DC optimal power flow. }
    \begin{lem} \label{lem_z*}
    	The optimal state $\mb{z}^*$ of the problem~\eqref{eq_dra2} satisfies the following: $\nabla_z H(\mb{z}^*) \in \mbox{span}(\mb{1}_N;-\mb{1}_m)$. 
    \end{lem}
\begin{proof}
	The proof directly follows from the KKT condition for convex cost function $H(\mb{z})$ and linear constraint~\eqref{eq_b}, see details in \cite{Boyd-CVXBook,bertsekas_lecture}.
\end{proof}

    To address the so-called box constraints in \eqref{eq_box}, smooth penalty terms and barrier functions are applied and added to the cost function. The new objectives are then modified by the following extra terms, 
    \begin{align} 
    	h_i  + \epsilon([z_i - \overline{d}]^+ + [\underline{d} - z_i ]^+)
    	\label{eq_fsigma}
    \end{align}
    \ab{where $\overline{d},\underline{d}$ are upper/lower-bound in \eqref{eq_box}},  $\epsilon \in \mathbb{R}^+$ is the weighting constant, and the penalizing function $[u]^+$ is in the form,
    \begin{align} 
    	[u]^+=\max \{u, 0\}^\sigma,~\sigma \in \mathbb{N}
    	\label{eq_sigma}
    \end{align}
    It is typical to consider smooth functions with $\sigma \geq 2$. The other typical smooth penalizing barrier function is in the form,
        \begin{align} 
    	[u]^+=\frac{1}{\sigma}\log(1+\exp(\sigma u)),~\sigma \in \mathbb{R}^+
    	\label{eq_sigma2}
    \end{align}
    It is known that this penalized case gets arbitrarily close to the exact optimizer $\mb{z}^*$ by choosing $\sigma$ sufficiently large \cite{nesterov1998introductory}.
    
    \section{The Proposed Distributed Nonlinear Solution } \label{sec_sol}
    \ab{Our proposed solution is based on gradient tracking which is a common approach to solve optimization problems. In distributed setup, on the other hand, some sort of (weighted) averaging the local gradients over the network is  needed to guide the decision states towards the global gradient direction. This gives a linear model in the form $\dot{z}_i = \eta a_i \sum_{j \in \mc{N}_i} W_{ij} (a_j\partial_{z_j} h_j -  a_i\partial_{z_i} h_i )$. To account for possible nonlinearity in the model,} our proposed \textit{nonlinear} distributed/decentralized\footnote{\ab{In this paper distributed and decentralized are used interchangeably. Our distributed/decentralized setup implies that each node solves its local optimization problem and shares relevant information with the neighboring nodes to coordinate and reach a global solution of the main global problem.} } Laplacian-gradient solution is in the form
    \begin{align} \label{eq_sol}
    	\dot{z}_i &= \eta a_i \sum_{j \in \mc{N}_i} W_{ij} g_n\Big(g_l(a_j\partial_{z_j} h_j) -  g_l(a_i\partial_{z_i} h_i)\Big)  
    \end{align}
    where $a_i \in \{+1,-1\}$ depending whether $i$ is generative node or reserving node.  \ab{If $i$ produces power then $a_i=+1$ and $i$ is a generative node. If $i$ is a battery-type and reserving power then $a_i=-1$ and $i$ is a reserving node}. This follows from Lemma~\ref{lem_z*}. A similar statement holds for $a_j$. $\eta \in \mathbb{R}^+$ is the step rate, $\mc{N}_i$ denotes the neighbours of node $i$, $W_{ij} \in \mathbb{R}^+$ is the weight on the link between $i$ and  $j$, $g_n(\cdot),g_l(\cdot) : \mathbb{R} \mapsto \mathbb{R}$ denote possible nonlinear mapping on the node or the link. If there is no nonlinearity, then $g_n(x)=x,g_l(x)=x$. \ab{Example possible nonlinearity $g_n(\cdot),g_l(\cdot)$ include: (i) saturation to address the so-called ramp-rate-limits in the automatic-generation-control setup, (ii) quantization to address quantized information channels for the information exchange among the energy nodes, or (iii) signum-based nonlinearities for fixed/finite-time (or prescribed-time) convergence and robust to impulsive noise. Note that these nonlinearities cannot be easily addressed in primal-dual-formulation solution (e.g., ADMM) and this is (along with addressing all-time feasibility and latency in Section~\ref{sec_delay}) motivates considering primal-based solution. Further, the existing ADMM-based solutions need all-time connectivity, while our proposed solution works under uniform connectivity }
    
    \begin{ass} \label{ass_nonlin}
    	The nonlinear mapping $g_n(\cdot),g_l(\cdot) : \mathbb{R} \mapsto \mathbb{R}$ are sign-preserving and odd. 
    \end{ass}
    \begin{ass} \label{ass_net}
	The network of energy nodes is uniformly connected (or B-connected) and undirected (i.e., $W_{ij}=W_{ji}$). 
   \end{ass}
   Initializing from a feasible solution $\sum_{i=1}^n a_i z_i(0) = b$ the following lemma holds. 
    \begin{lem}
    	The solution by \eqref{eq_sol} is all-time feasible under Assumption~\ref{ass_nonlin} and \ref{ass_net}.
    \end{lem}
    \begin{proof}
 	\ab{Recall that all-time feasibility implies that initializing from a feasible solution $\sum_{i=1}^n a_i z_i(0) = b$, the constraint $\sum_{i=1}^n a_i z_i(t) = b$ must hold at all times, i.e., we need to prove that the change is zero $\sum_{i=1}^n a_i \dot{z}_i = 0$. Finding this term from Eq.~\eqref{eq_sol}, } 
 	\begin{align} \nonumber
 		\sum_{i=1}^n a_i \dot{z}_i = \sum_{i=1}^n \eta a_i^2 \sum_{j \in \mc{N}_i} W_{ij} g_n\Big(&g_l(a_j\partial_{z_j} h_j) \\
 		&-  g_l(a_i\partial_{z_i} h_i)\Big)  \label{eq_feas_proof}
 	\end{align}
    From Assumption \ref{ass_net} we have $W_{ij}=W_{ji}$ and from Assumption \ref{ass_nonlin} we have 
  	\begin{align} \nonumber
 	g_n\Big(g_l(a_j\partial_{z_j} h_j) &-  g_l(a_i\partial_{z_i} h_i)\Big) \\&= -g_n\Big(g_l(a_i\partial_{z_i} h_i) -  g_l(a_j\partial_{z_j} h_j)\Big) \nonumber
    \end{align}
   \ab{This follows from the assumption that for every $j \in \mc{N}_i$ we have $i \in \mc{N}_j$, and from the above equation the summation \eqref{eq_feas_proof} over all $i,j$ is equal to zero.} Therefore, initializing a feasible solution, we have  
    	\begin{align} \nonumber
   	\sum_{i=1}^n a_i z_i(t) = \sum_{i=1}^n a_i z_i(0) = b.  
   \end{align}
   and the proof follows.  
    \end{proof} 
	First, note that the optimal point satsifying $\nabla_z H(\mb{z}^*) \in \mbox{span}(\mb{1}_N;-\mb{1}_m)$ (as described in Lemma~\ref{lem_z*}) is invariant under dynamics \eqref{eq_sol}. The next theorem proves convergence to this optimal point. First, we provide a lemma to help prove the theorem.
	\begin{lem} \label{lem_sum}
		Under Assumption~\ref{ass_net} and nonlinear mapping $g:\mathbb{R} \mapsto \mathbb{R}$ satisfying Assumption~\ref{ass_nonlin}, for $\mb{z} \in \mathbb{R}^n$,
		\begin{align} \label{eq_sum_lem}
			\sum_{i=1}^n z_i \sum_{j=1}^n W_{ij} g(z_j-z_i) = \sum_{i,j=1}^n \frac{W_{ij}}{2} (z_j-z_i) g(z_j-z_i)
		\end{align}
	\end{lem}
	\begin{proof}
		We have $W_{ij} = W_{ji}$ under Assumption~\ref{ass_net} and  $g(z_j-z_i) = -g(z_i-z_j)$ under Assumption~\ref{ass_nonlin}. Thus, 
		\begin{align} \nonumber
			z_i W_{ij} g(z_i-z_j) + &z_j W_{ji}  g(z_i-z_j) \\ \nonumber
			& = W_{ij}(z_i-z_j) g(z_j-z_i) \\
			& = W_{ji}(z_j-z_i) g(z_i-z_j).
		\end{align}
		and the proof follows. 
	\end{proof}
   \begin{theorem}
	   Under Assumption~\ref{ass_nonlin} and \ref{ass_net} and with feasible initialization, the dynamics \eqref{eq_sol} converges to the optimal solution of problem \eqref{eq_dra2}.
   \end{theorem} 
\begin{proof}
	Consider the Lyapunov fucntion  $\overline{H} := \sum_{i=1}^n h_i(z_i)-\sum_{i=1}^n h_i(z^*_i)$ (as the residual cost). Then, having $\dot{\overline{H}} = \nabla_z H^\top \dot{\mb{z}}$ under dynamics \eqref{eq_sol},
   \begin{align} \nonumber
   	\dot{\overline{H}}
   	= \sum_{i =1}^n \partial_{z_i} h_i\eta a_i \sum_{j \in \mc{N}_i} W_{ij} g_n\Big(&g_l(a_j\partial_{z_j} h_j) \\ &-  g_l(a_i\partial_{z_i} h_i)\Big) \label{eq_proof1}
   \end{align}
 Then, following from Lemma~\ref{lem_sum},
   \begin{align} \nonumber
	\dot{\overline{H}}
	= \sum_{i,j =1}^n  \eta a_i  W_{ij} (\partial_{z_i} h_i- \partial_{z_j} h_j) g_n\Big(&g_l(a_j\partial_{z_j} h_j) \\ &-  g_l(a_i\partial_{z_i} h_i)\Big) \label{eq_proof2}
\end{align}
 Under Assumption~\ref{ass_nonlin}, we have $\dot{\overline{H}}\leq 0$ and the proof follows from Lyapunov theory \cite{nonlin}.
\end{proof}

As a special case, one can add sgn-based (or sign-based) nonlinearity to improve the convergence rate in finite-time or fixed-time \cite{shang2017finite,taes2020finite,polyakov2011nonlinear,song2021fixed}. Then, one can choose $g_l(u)$ and/or $g_n(u)$ as $\mbox{sgn}^{\mu_1}(u) + \mbox{sgn}^{\mu_2}(u)$ where  $\mbox{sgn}^\mu(u)=\frac{u^{\mu}}{|u|}$ with  $0<\mu_1<1$ and  $\mu_2>1$. It can be shown that such solutions reach faster convergence which is also verified in the simulation section of this paper. Moreover, such sgn-based schemes motivate robust-to-impulsive-noise solutions \cite{zayyani2016distributed}.

\section{Solution under Time-Delays} \label{sec_delay}
\ab{We consider the discrete-time version of the solution by replacing $\dot{z}_i = \frac{z_i(k+1) -z_i(k)}{T} $ and $\eta_\tau = \eta T$ in~\eqref{eq_sol}}. For this case, we assume no nonlinearity on the nodes, i.e., $g_n(x)=x$, and we get\footnote{\ab{This assumption is to satisfy all-time feasibility of the solution in the presence of time-delays. In case there exist both node-nonlinearity and time-delays it is not easy to satisfy the constraint $\sum_{i=1}^n a_i z_i(t) = b$ and other types of solutions must be considered. }}

\small
\ab{\begin{align} \nonumber
	z_i(k+1) = z_i(k) +\eta_\tau a_i &\sum_{j \in \mc{N}_i}  W_{ij}  \Big(g_l(a_j\partial_{z_j} h_j(k)) \\&-g_l(a_j\partial_{z_i} h_i(k))\Big),
	\nonumber 
\end{align}\normalsize
Next, assume maximum $\overline{\tau}$ steps of time-delays in the data-transmission network among the energy nodes. Then, the solution is in the form  
}

\small
\begin{align} \nonumber
	z_i(k+1) = z_i(k) +\eta_\tau a_i &\sum_{j \in \mc{N}_i}  \sum_{r=0}^{\overline{\tau}} W_{ij}  \Big(g_l(a_j\partial_{z_j} h_j(k-r)) \\&-g_l(a_j\partial_{z_i} h_i(k-r))\Big) \mc{I}_{k-r,ij}(r),
	\label{eq_sol_delay} 
\end{align}\normalsize
with $\mc{I}$ as the indicator function,
\begin{align} \label{eq_I}
	\mc{I}_{k,ij}(\tau) = \left\{ \begin{array}{ll}
		1, & \text{if}~  \tau_{ij}(k) = \tau,\\
		0, & \text{otherwise}.
	\end{array}\right.
\end{align}
\ab{In the summation in Eq.~\eqref{eq_sol_delay}, this indicator function is $1$ if the delay at step $k$ is $\tau$ and $0$ for other values. This simply implies that only the received information from node $j$ with delay $\tau$ contributes in $\sum_{j \in \mc{N}_i} $. }
The assumption on the time-delay is as follows:
\begin{ass} \label{ass_delay} 
The time-delay over link between $i,j$ at step $k$ is $\tau_{ij}(k) \leq \overline{\tau}$. Max delay $\overline{\tau}$  ensures that message of $i$ at step $k$ eventually reaches $j$ (at most) at step $k+\overline{\tau}$. $\tau_{ij}(k)$ is heterogeneous, arbitrary, possibly time-variant, and symmetric over undirected links. 
\end{ass}
Note that, the above assumption is a general form of time-delays considered in the literature, e.g., see \cite{hadjicostis2013average}. Heterogeneity implies that the delays at different links are in general different. Time-variance implies that the delays may change over time but remain bounded by $\overline{\tau}$. 
For static energy units in constant distance from each other, the delays are proportional to the distances \cite{liu2019survey}, which justifies symmetric delays over undirected links.  

To prove convergence first recall the following lemmas.
\begin{lem} \label{lem_strict}
	For a strictly-convex function $H:\mathbb{R}^n \mapsto \mathbb{R}$ with  $2 v <  \frac{d^2 h_i(z_i)}{dz_i^2} < 2 u$, ${\delta(k) := \mb{z}(k+1)-\mb{z}(k)}$, 
	\begin{align} \label{eq_taylor_1}
		H(\mb{z}(k+1)) \geq H(\mb{z}(k)) + \nabla_z H(\mb{z}(k))^\top \delta(k) +  v\delta(k)^\top \delta(k)
		\\
		H(\mb{z}(k+1)) \leq H(\mb{z}(k)) + \nabla_z H(\mb{z}(k))^\top \delta(k) +  u\delta(k)^\top \delta(k) 
		\label{eq_taylor_2}
	\end{align}
\end{lem}
\begin{proof}
	The proof of the above lemma is given in general optimization handbooks, e.g., see \cite{Boyd-CVXBook}.
\end{proof} 
For the network connecting the energy nodes define its Laplacian $L=D-W$ with the diagonal degree matrix $D$ defined as $D = \mbox{diag}[\sum_{j=1}^n W_{ij}]$.
\begin{lem}  \label{lem_xLy}
	For an undirected network, its Laplacian matrix $L$ is positive semi-definite. Let ${\overline{\mb{z}} := \mb{z} - \frac{\mb{1}_n^\top \mb{z}}{n} \mb{1}_n}$, and  $\lambda_n,\lambda_2$ as the largest and smallest non-zero eigenvalue of $L$. Then, 
	\begin{align}    \label{eq_laplace}
		\lambda_2 \|\overline{\mb{z}} \|_2^2 &\leq \mb{z}^\top L \mb{z} = \overline{\mb{z}}^\top L \overline{\mb{z}} \leq \lambda_n \|\overline{\mb{z}} \|_2^2
	\end{align}
Further, given sign-preserving odd nonlinearity $g$ such that $\kappa \leq \frac{g(z_i)}{z_i} \leq \mc{K}$ with $\kappa,\mc{K} \in \mathbb{R}^+$, 
\begin{align} \label{eq_laplace2}
	\lambda_2 \kappa \|\overline{\mb{z}} \|_2^2\leq g(\mb{z})^\top L \mb{z} \leq \lambda_n \mc{K} \|\overline{\mb{z}} \|_2^2 
\end{align}
\end{lem}
\begin{proof}
	The proof of \eqref{eq_laplace} is given in standard consensus literature, e.g., see \cite{olfatisaberfaxmurray07}. For the second part, 
	\begin{align} \nonumber
		g(\mb{z})^\top L \mb{z} &= \overline{g(\mb{z})}^\top L \overline{\mb{z}} \\ \label{eq_proof_Ls}
		&=  \frac{1}{2}\sum_{i,j=1}^n W_{ij}(g(z_i)-g(z_j))(z_i-z_j)
	\end{align}
	with symmetric matrix $W$ and ${\overline{g(\mb{z})} := g(\mb{z}) - \frac{\mb{1}_n^\top g(\mb{z})}{n} \mb{1}_n}$. 
	Following the monotonic property of $g(\mb{z})$,
	\begin{align}\nonumber
		\kappa (z_i-z_j)(z_i-z_j) \leq  &(g(z_i)-g(z_j))(z_i-z_j) \\
		&\leq \mc{K}(z_i-z_j)(z_i-z_j)
	\end{align} 
	Using the above in \eqref{eq_proof_Ls} along with \eqref{eq_laplace} proves \eqref{eq_laplace2}.
\end{proof}
Recall that the second largest eigenvalue of the network Laplacian $\lambda_2$ plays a key role in the convergence of the consensus algorithms \cite{olfatisaberfaxmurray07}. It is known that this value is directly related to the network connectivity and thus is referred to as algebraic connectivity \cite{godsil}.
Using the above lemmas, the next theorem proves the convergence. Recall that with some abuse of notation and assuming no latency ($\overline{\tau}=0$), we rewrite the solution~\eqref{eq_sol_delay} as
 \begin{align} 
 	\mb{z}(k+1) = \mb{z}(k) - \eta_\tau \mc{A}  L \varphi \label{eq_L}
 \end{align}
with $\mc{A}:= \mbox{diag}[a_1,\dots,a_N]$ and $\varphi := g_l(\mc{A}\nabla_z H)$.
\begin{theorem} \label{thm_delay}
	Under Assumption~\ref{ass_nonlin},~\ref{ass_net}, and ~\ref{ass_delay}, the solution \eqref{eq_sol_delay} converges to the optimizer of \eqref{eq_dra2} for  
	\begin{align} \label{eq_eta2}
		\eta_\tau <  \frac{\kappa \lambda_2}{u \lambda_n^2 \mc{K}^2(\overline{\tau}+1)} 
	\end{align}
    with $\lambda_n,\lambda_2$ as the largest and smallest non-zero eigenvalue of $L$ and  $\kappa,\mc{K}$ as the sector bounds of $g_l$ (as described in Lemma~\ref{lem_xLy}). 
\end{theorem}

\begin{proof}
	We first prove the bound on $\eta_\tau$ for discrete-time convergence in the absence of delay. Then extend the results to the presence of time-delay.
	Consider discrete Lyapunov-type residual function $\overline{H}(k) := H(\mb{z}(k))-H(\mb{z}^*)$. We prove $\overline{H}(k+1) < \overline{H}(k)$ under dynamics~\eqref{eq_sol_delay}  for $\mb{z}(k)  \neq \mb{z}^*$. Define ${\delta(k) := \mb{z}(k+1)-\mb{z}(k)}$.
	To satisfy $\overline{H}(k+1) \leq \overline{H}(k)$, from Lemma~\ref{lem_strict},
	\begin{align} \label{eq_proof_1}
		\nabla_z H^\top \delta(k)  + u \delta(k)^\top \delta(k)  \leq 0.
	\end{align} 
	Under  \eqref{eq_sol_delay}, 
	\begin{align} \label{eq_proof_2}
		-\eta_\tau \mc{A}  \nabla_z H^\top     L\varphi  + u \mc{A}^2 \eta_\tau^2     \varphi^\top L^\top L\varphi \leq 0;
	\end{align} 
	Define $\xi(\mb{z}) := \mc{A}\nabla_z H -  \frac{\mc{A}}{n}\sum_{i=1}^n \frac{ dh_i(z_i)}{dz_i} \mb{1}_n$. Note that $\mc{A}^2=I_n$. From Assumption~\ref{ass_nonlin} and Lemma~\ref{lem_xLy}, the above is satisfied for,
	\begin{align} 
		\label{eq_proof__rho}
		(-\kappa \eta \lambda_2 + u \lambda_n^2 \mc{K}^2 \eta^2)    \xi^\top    \xi \leq 0,
	\end{align}
	with the strict inequality for
	\begin{align} \label{eq_eta}
		\eta_\tau <  \frac{\kappa \lambda_2}{u \lambda_n^2 \mc{K}^2} =: \overline{\eta}
	\end{align}
   The above implies that for $\eta_\tau < \overline{\eta}$ the residual function is decreasing \ab{and state variable $\mb{z}$ converges to the equilibrium $\mb{z}^*$ and $H(\mb{z})$ converges to the optimal cost $H(\mb{z}^*)$.} Next, we prove convergence under latency. 	   
   For general heterogeneous time-varying delays, $\delta(k)$ needs to be scaled by $\overline{\tau}+1$ and, thus, $\eta_\tau$ needs to be down-scaled by $\overline{\tau}+1$ to ensure convergence. 
   For time-varying delays, from~\eqref{eq_I},
   $0 <  \sum_{r=0}^{\overline{\tau}} \mc{I}_{k-r,ij}(r) \leq (\overline{\tau}+1)$.
   This implies that $\delta(k)$ is scaled by $\overline{\tau}+1$ and following the same procedure as in the proof, $\eta_\tau$ needs to be down-scaled by $\overline{\tau}+1$, i.e., $ \eta_\tau (\overline{\tau}+1) < \overline{\eta}$.
   Therefore, in general, $\eta_\tau$ in the following ensures convergence in the presence of latency,
   \begin{align} \nonumber
	\eta_\tau <  \frac{\kappa \lambda_2}{u \lambda_n^2 \mc{K}^2(\overline{\tau}+1)} = \frac{\overline{\eta}}{\overline{\tau}+1}
   \end{align}   
    This proves the theorem.
\end{proof}
\ab{From Eq.~\eqref{eq_eta2} one can decrease the step-rate $\eta_\tau$ to handle larger max delay $\overline{\tau}$.} Note that the bound on the step rate $\eta_\tau$ is a function of $\lambda_2,\lambda_n$ which are in turn a function of network properties (e.g., its connectivity and size). Also, the sector-bounds $\kappa,\mc{K}$ of the nonlinear mapping affect the bound on $\eta_\tau$. In case of homogeneous delays $\tau_{ij}=\overline{\tau}$ at all energy nodes, i.e., the nodes' states at any time-step $k$ get updated based on the information at $k-\overline{\tau}-1$. Similarly, for this case assuming a longer time-scale of $\overline{\tau}+1$ steps ensures the convergence, i.e., by down-scaling $\eta_\tau$ by $\overline{\tau}+1$. Thus, the same bound \eqref{eq_eta2} ensures convergence. 
The algorithm for energy resource management under communication time-delay is summarized in Algorithm~\ref{alg_1}, assuming that the given $\eta_\tau$ satisfies \eqref{eq_eta2} for the given max $\overline{\tau}$.

\begin{algorithm} 
	\textbf{Input:}  $\mc{N}_i$, $W$, $\eta_\tau$, $\overline{d}$, $\underline{d}$, $b$, $h_i(\cdot)$, $\overline{\tau}$\;
	\textbf{Initialization:} $k=0$, energy node $i$ chooses a feasible initial state\;
	\While{termination criteria NOT true}{
		Node $i$ receives a (possibly delayed) packet including $\partial_{z_j} h_j$ from nodes in $j \in \mc{N}_i$\;
		Node $i$ computes Eq.~\eqref{eq_sol_delay}\;
		Node $i$ shares the updated $\partial_{z_i} h_i$ with neighboring nodes $i \in \mc{N}_j$\;
		$k \leftarrow k+1$\;
	}
	\textbf{Return} Final state $z_i$ and cost $h_i(z_i)$\;	
	\caption{Energy Resource Management under Communication Time-Delay} 
	\label{alg_1}
\end{algorithm}

\section{Simulations} \label{sec_sim}
For simulation we consider the quadratic cost model as in \cite{kar_edp}. The cost $f_i$ and $e_i$ are defined in the form
\begin{align} \label{eq_f_quad}
 \gamma_i x_i^2 + \beta_i x_i + \alpha_i 
\end{align}
We consider different ranges of parameters $\alpha_i,\beta_i,\gamma_i$ in this section. First, for the sake of comparison, we consider similar parameters as in \cite{kar_edp}. The feasibility constraint is $b=700 MW$, i.e., 
\begin{align} \label{eq_b_sim}
	\sum_{i=1}^N x_i  = 700 + \sum_{i=1}^m y_i,
\end{align}

First, consider a cycle network of $N=n=10$ all generator power units to distributedly optimize their costs \eqref{eq_f_quad} under the feasibility constraint \eqref{eq_b_sim} and box constraints \eqref{eq_box} with $\underline{d}=20$  and $\overline{d} =90$ for generators. Solution under different linear and nonlinear scenarios are compared in Fig.~\ref{fig_sim1}. For the nonlinearity, we chose the saturation function and the sgn-based function. The accelerated linear solution is from the work \cite{shames2011accelerated} with parameter $\overline{\beta} = 0.5$. For the sgn-based solution, the evolution of power states is shown along with their weighted average $\frac{1}{n}\sum_{i=1}^n {a_i z_i} = \frac{b}{n}$  to check the feasibility. For the simulation, the step rate is $\eta =1$. 
 
\begin{figure}[]
	\centering
	\includegraphics[width=2.5in]{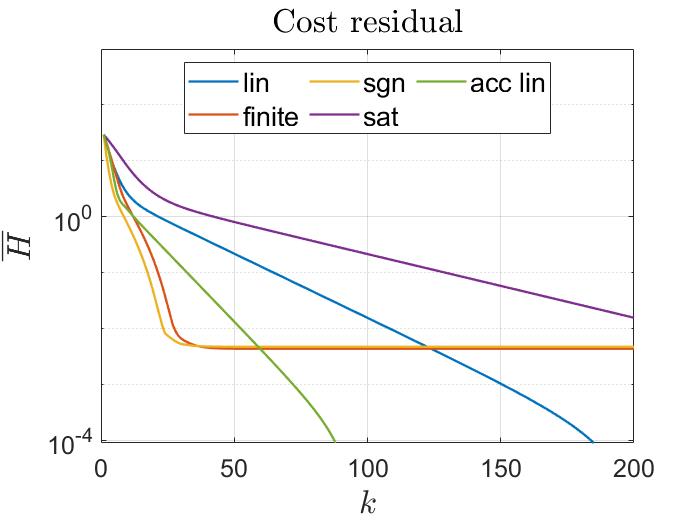}
	\includegraphics[width=2.5in]{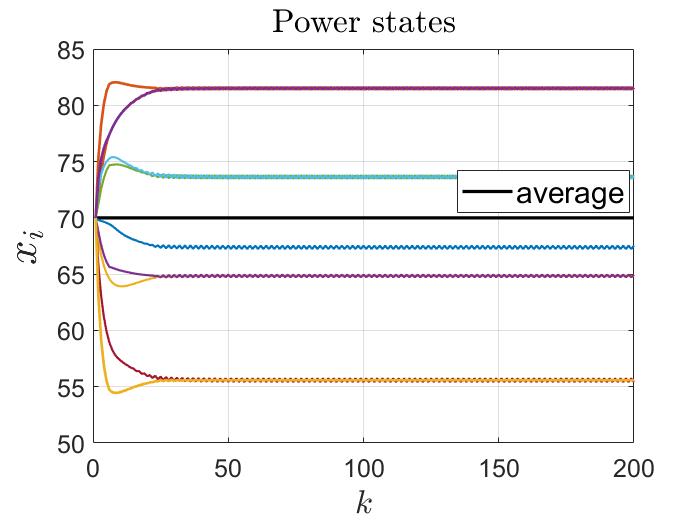}
	\caption{(Top) The evolution of the cost residuals is compared under different linear and nonlinear solutions. The sgn-based solution converges faster than others. (Bottom) The generated power states under a nonlinear sgn-based solution are shown. The average (black line in the middle) remains constant implying that the resource-demand feasibility holds. 
	} 
	\label{fig_sim1}
\end{figure}

Next, we consider the case that some of the energy units reserve energy (battery) and some produce energy (generator). Consider a 2-hop cycle network of $n=10$ energy nodes with $N=7$ generators and $m=3$ battery nodes. The cost model for the batteries is linear in the form 
\ab{
\begin{align} \label{eq_f_lin}
	\beta_i y_i + \alpha_i 
\end{align}}
The box constraints are $\underline{d}=20$ for generators, $\underline{d}=0$ for batteries,  and upper limit $\overline{d} =200$ for all. Under initialization $x_i(0)=100MW$ and $y_i(0)=0MW$, the power states under sgn-based solution with $\mu_1=0.5$ and $\mu_2=1.1$ are illustrated in Fig.~\ref{fig_sim2}. The cost residuals under linear and different sgn-based solutions are compared in Fig.~\ref{fig_sim2}. As it is clear from the figure, by adding sgn-based nonlinearity in the solution the convergence rate is improved as compared with the linear solution.
\begin{figure}[]
	\centering
	\includegraphics[width=2.5in]{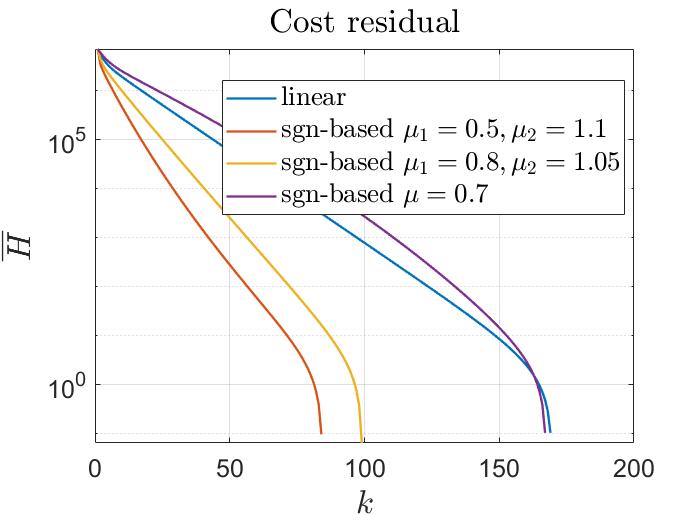}
	\includegraphics[width=2.5in]{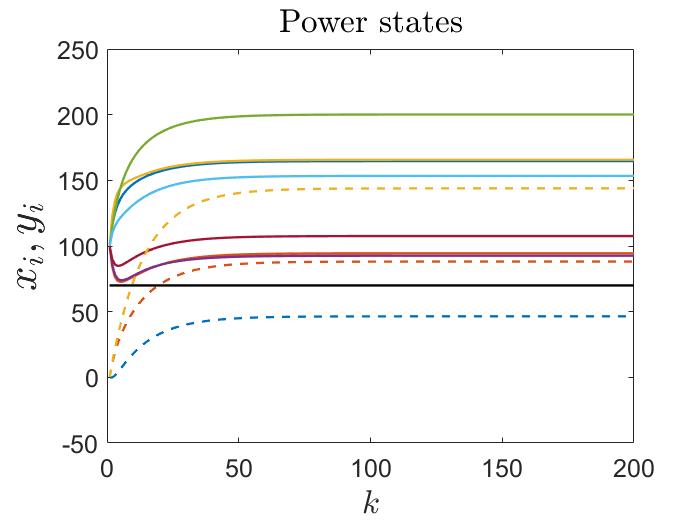}
	\caption{(Top) Cost residuals for generators plus batteries converge under different linear and nonlinear solutions. (Bottom) The generated power states (solid lines) and reserved powers (dashed lines) under the nonlinear sgn-based solution reach steady-state value while the weighted average (black solid line) remains constant implying resource-demand feasibility. 
	} 
	\label{fig_sim2}
\end{figure}

Next, we simulate the solution under communication latency (over a 2-hop cycle), i.e., the information sent over the network reaches the destination with some bounded time-delay. The delays at the links are different (implying heterogeneity) but, in terms of time-dependency, we consider two cases: time-varying and time-invariant communication delays. For the first case, we randomly generate an integer number between $0$ and $\overline{\tau}$ denoting the delay at every iteration, and for the second case the randomly-generated delay remains constant over the time-evolution of the solution dynamics. The solutions for different max values of the time-delays $\overline{\tau}$ are shown in Fig.~\ref{fig_sim3}. The other cost and state parameters are chosen randomly as in the previous simulations. It is clear that for the given step-rate by increasing the max delay $\overline{\tau} \geq 5$ (time-varying delays) and $\overline{\tau} \geq 2$ (time-invariant delays) the solution loses stability and convergence. Recall that, from the proof of Theorem~\ref{thm_delay}, a remedy to tolerate higher time-delays (i.e., for larger $\overline{\tau}$) is to reduce the discrete-time step-rate $\eta_\tau$. \ab{Next we decrease step-rate $\eta_\tau$ to half of its previous value to tolerate larger time-delays. For the new step-rate the residual simulations are re-produced and for this case, as it is evident from Fig.~\ref{fig_sim4}, the solution can \textit{tolerate} the max delay $\overline{\tau} \leq 5$ for time-varying delays and $\overline{\tau} \leq 3$ for time-invariant delays. We further simulate the evolution of the power states for time-varying delays with $\overline{\tau} = 5$ for large and small step rates. Fig.~\ref{fig_sim5} shows that how reducing the step rate works as a remedy to tackle high time-delays. It is worth mentioning that although the solution is unstable for large step-rates and  large time-delay, it still preserves its all-time feasibility (the average remains constant while the states are oscillating).}

\begin{figure}[]
	\centering
	\includegraphics[width=2.5in]{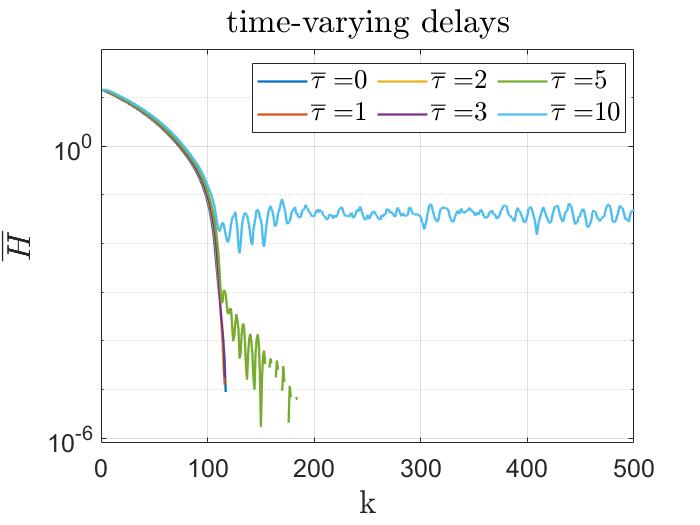}
	\includegraphics[width=2.5in]{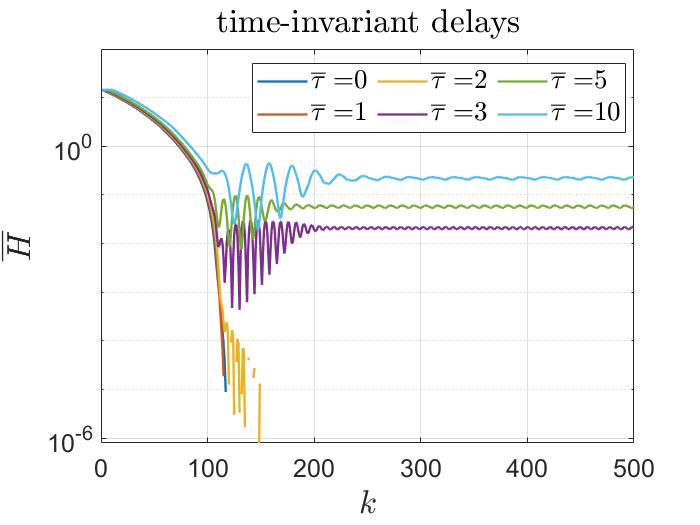}
	\caption{Cost residuals under different time-delayed scenarios: (Top) time-varying heterogeneous delays (Bottom) time-invariant heterogeneous delays. 
	} 
	\label{fig_sim3}
\end{figure}

\begin{figure}[]
	\centering
	\includegraphics[width=2.5in]{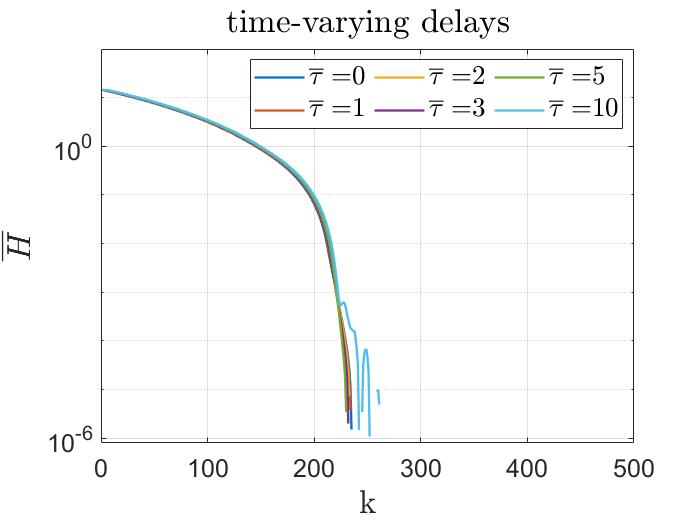}
	\includegraphics[width=2.5in]{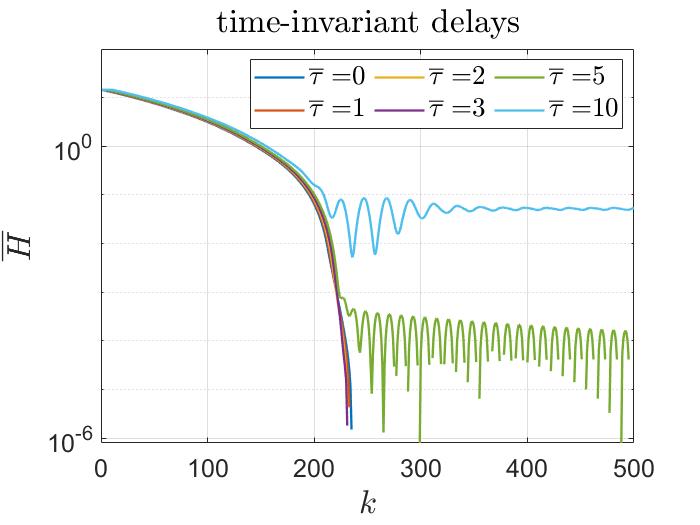}
	\caption{Time-delayed cost residuals under reduced step-rate: (Top) time-varying heterogeneous delays (Bottom) time-invariant heterogeneous delays. 
	} 
	\label{fig_sim4}
\end{figure}

\begin{figure}[]
	\centering
	\includegraphics[width=2.5in]{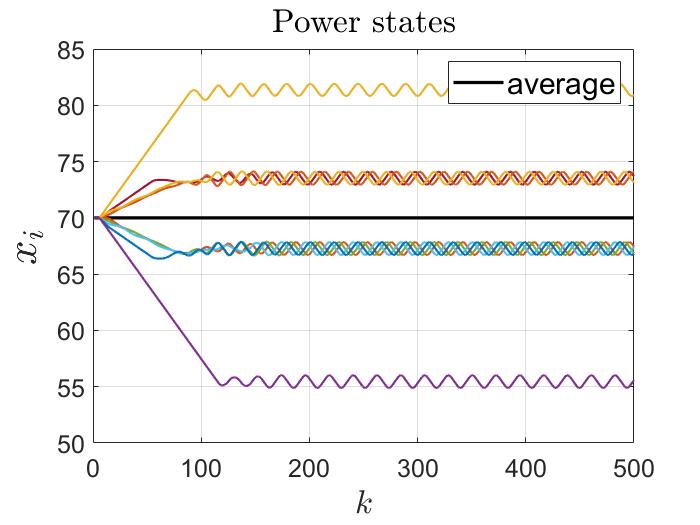}
	\includegraphics[width=2.5in]{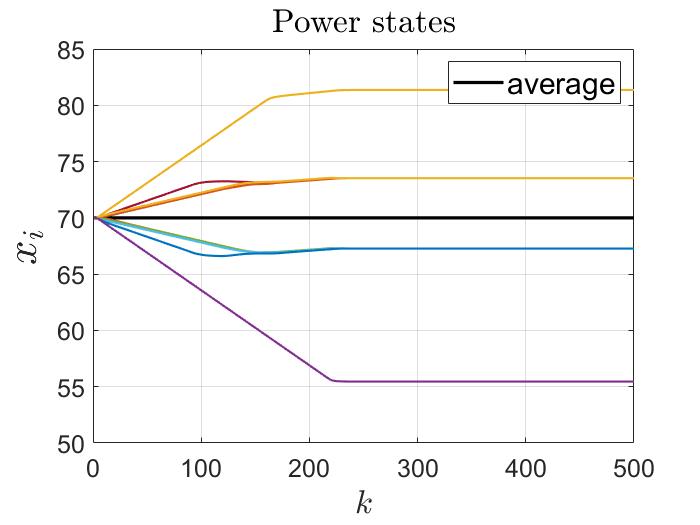}
	\caption{\ab{The figure shows the evolution of states under time-varying delays with $\overline{\tau} = 5$ for $\eta_\tau =0.5$ (top) and $\eta_\tau =0.25$ (bottom). Reducing the step rate enables the algorithm to tolerate higher time-delay values.} 
	} 
	\label{fig_sim5}
\end{figure}

\section{Conclusions and Future Directions} \label{sec_con}
\subsection{Concluding Remarks}
This work studies distributed algorithms for energy resource-management over a network of energy nodes, both generating-type and reserving-type nodes. The solution is proposed both in continuous-time, and discrete-time under communication time-delay. The proposed Laplacian-gradient-based solution addresses all-time resource-demand feasibility, possible time-delay in the communication network of energy nodes, and nonlinear dynamics due to possible sgn-based or other inherent model nonlinearities. 

\subsection{Future Research}
As one direction of future research, one can consider the communication network under packet drop or link failure \cite{icrom}. In other words, the idea is to check under what conditions over an unreliable communication network the feasibility and convergence hold. Addressing security concerns with privacy-preserving algorithms is another direction of future research. In this case one can use secure consensus algorithms \cite{mo2016privacy} for message passing over the communication network.

	\bibliographystyle{IEEEbib}
	\bibliography{bibliography}
\end{document}